\documentclass[11pt]{article}
\usepackage[margin=1in]{geometry}
\usepackage{amsmath,amssymb,amsthm}
\usepackage{aliascnt}
\usepackage{microtype}
\usepackage[hidelinks]{hyperref}
\usepackage{xcolor} 

\newtheorem{theorem}{Theorem}[section]
\newaliascnt{lemma}{theorem}
\newtheorem{lemma}[lemma]{Lemma}
\aliascntresetthe{lemma}
\newaliascnt{proposition}{theorem}
\newtheorem{proposition}[proposition]{Proposition}
\aliascntresetthe{proposition}
\newaliascnt{corollary}{theorem}
\newtheorem{corollary}[corollary]{Corollary}
\aliascntresetthe{corollary}

\newcommand{\cM}{\mathcal M}
\newcommand{\Pf}{\operatorname{Pf}}

\newcommand{\crs}{\operatorname{cr}}
\newcommand{\rev}{\operatorname{rev}}
\newcommand{\coRP}{\mathsf{coRP}}
\newcommand{\coNP}{\mathsf{coNP}}

\newcommand{\x}{\mathbf{x}}

\title{Pfaffian Recognition is in $\mathsf{coRP}$}

\author{
Priyanshu Pant\thanks{Email: \texttt{priyanshupant03@gmail.com}}
\qquad
Ranveer Singh\thanks{Email: \texttt{ranveer@iiti.ac.in}}\\[0.5ex]
Department of Computer Science and Engineering\\
Indian Institute of Technology Indore
}

\date{}

\begin{document}

\maketitle

\begin{abstract}
An orientation of a graph is called Pfaffian if all perfect matchings occur with the same sign in its Pfaffian expansion. A graph is Pfaffian if it admits such an orientation. Given a Pfaffian orientation, the number of perfect matchings can be computed in polynomial time. In 1963, Kasteleyn showed that every planar graph is Pfaffian. In 1989, Vazirani and Yannakakis showed that deciding whether a graph is Pfaffian is polynomial time equivalent to testing whether a given orientation is Pfaffian. For bipartite graphs, polynomial time recognition algorithms were obtained by Robertson, Seymour, and Thomas in 1999, and independently by McCuaig in 2004. Whether a deterministic polynomial time algorithm exists for general graphs remains open.

In this paper, we first construct a polynomial that is identically zero if and only if a given orientation is Pfaffian, and show that this polynomial can be evaluated in polynomial time. Combined with the reduction of Vazirani and Yannakakis, this gives a randomized polynomial time algorithm for deciding whether a general graph is Pfaffian, placing the problem in \(\mathsf{coRP}\). We also prove that an orientation of a matching covered graph is non-Pfaffian if and only if reversing the edges of some cut makes it strongly connected. Although deciding whether such a cut exists is NP-complete for general digraphs, our characterization yields a randomized polynomial time algorithm for matching covered graphs and deterministic polynomial time algorithms for bipartite and planar matching covered graphs.
\end{abstract}

\section{Introduction}
\label{sec:introduction}

A \emph{perfect matching} of a graph is a set of pairwise disjoint edges that covers every vertex exactly once. Counting perfect matchings is a central problem in combinatorics and counting complexity~\cite{LucchesiMurty2024,LovaszPlummer1986}. Although the existence of a perfect matching can be decided in polynomial time~\cite{Edmonds1965}, Valiant proved that counting perfect matchings is \(\#\mathsf{P}\)-complete, even for bipartite graphs~\cite{Valiant1979}.

All graphs considered in this paper are finite and simple. A graph is \emph{matching covered} if it is connected, has at least one edge, and every edge belongs to some perfect matching. For a graph with a perfect matching, deleting all edges that belong to no perfect matching leaves a disjoint union of matching covered graphs without changing its set of perfect matchings. Thus, it suffices to consider matching covered graphs when counting perfect matchings. For a graph \(G\), we write \(V(G)\) and \(E(G)\) for its vertex and edge sets, and \(\mathcal M(G)\) for its set of perfect matchings. Let \(G\) have \(2n\) vertices, and assign a variable \(x_e\) to each edge \(e\in E(G)\). The \emph{perfect matching polynomial} of \(G\) is
\[
    \operatorname{PerfMatch}(G) = \sum_{M\in\mathcal M(G)} \prod_{e\in M}x_e .
\]
Setting all variables \(x_e=1\) gives the number of perfect matchings of \(G\).

A classical algebraic approach to counting perfect matchings uses Pfaffian orientations. An \emph{orientation} \(D\) of \(G\) is a choice of direction for each edge. Fix an ordering \(1,2,\ldots,2n\) of the vertices of \(G\). If \(e=ij\in E(G)\), we write \(x_e\) and \(x_{ij}\) interchangeably for the corresponding edge variable. The \emph{weighted skew-adjacency matrix} of \(G\) with respect to \(D\), denoted by \(A_D\), is the \(2n\times2n\) matrix whose \((i,j)\)-entry is
\[
    (A_D)_{ij} = \begin{cases}
        x_{ij},&\text{if }ij\in E(G)\text{ is oriented from }i\text{ to }j,\\
        -x_{ij},&\text{if }ij\in E(G)\text{ is oriented from }j\text{ to }i,\\
        0,&\text{if }ij\notin E(G).
    \end{cases}
\]
Let \(M=\{e_1,e_2,\ldots,e_n\}\) be a perfect matching of \(G\), and suppose that in the orientation \(D\), the tail and the head of \(e_i\) are \(u_i\) and \(v_i\), respectively. The \emph{permutation of \(M\) in \(D\)}, denoted by \(\pi_D(M)\), is
\[
    \pi_D(M) = \begin{pmatrix}
        1&2&3&4&\cdots&2n-1&2n\\
        u_1&v_1&u_2&v_2&\cdots&u_n&v_n
    \end{pmatrix}.
\]
The \emph{sign of \(M\) with respect to \(D\)}, denoted by \(\operatorname{sign}_D(M)\), is the sign of \(\pi_D(M)\). This value is independent of the order in which the edges of \(M\) are enumerated. The Pfaffian of \(A_D\) is
\[
    \operatorname{Pf}(A_D) = \sum_{M\in\mathcal M(G)} \operatorname{sign}_D(M) \prod_{e\in M}x_e .
\]
It contains exactly the same monomials as \(\operatorname{PerfMatch}(G)\), with signs determined by the orientation. If all perfect matchings have the same sign with respect to \(D\), then
\[
    \operatorname{Pf}(A_D) = \pm\operatorname{PerfMatch}(G).
\]
Such an orientation is called a \emph{Pfaffian orientation}, and a graph admitting one is called \emph{Pfaffian}. Once a Pfaffian orientation is known, the number of perfect matchings can be computed in polynomial time using Cayley's identity~\cite{Cayley1849},
\[
    \det(A)=\operatorname{Pf}(A)^2,
\]
which holds for every skew-symmetric matrix \(A\) of even order, where \(\operatorname{Pf}(A)\) is defined by viewing \(A\) as the weighted skew-adjacency matrix of the underlying graph. This motivates the \emph{Pfaffian recognition problem}, which asks whether a given graph admits a Pfaffian orientation.

Pfaffian orientations originate in Kasteleyn's work on the dimer problem. In 1963, Kasteleyn showed that every planar graph admits a Pfaffian orientation and gave a polynomial time algorithm for finding one, yielding a polynomial time algorithm for counting perfect matchings in planar graphs~\cite{Kasteleyn1963,Kasteleyn1967}. Since Pfaffian graphs need not be planar, characterizing them and finding Pfaffian orientations efficiently have been extensively studied, as such orientations allow perfect matchings to be counted in polynomial time~\cite{Thomas2006,LucchesiMurty2024}.

For bipartite graphs, deciding whether a Pfaffian orientation exists is equivalent to several classical problems, including P\'olya's permanent problem, recognizing sign-nonsingular matrices, and detecting even directed cycles~\cite{RST1999,McCuaig2004,Thomas2006}. In 1975, Little~\cite{Little1975} proved that a bipartite graph with a perfect matching is Pfaffian if and only if it contains no even subdivision of \(K_{3,3}\) whose deletion leaves a graph with a perfect matching. In 1999, Robertson, Seymour, and Thomas~\cite{RST1999}, and independently McCuaig in 2004~\cite{McCuaig2004}, obtained structural characterizations which lead to polynomial time algorithms for deciding whether a bipartite graph is Pfaffian. Björklund, Husfeldt, and Kaski~\cite{BjorklundHusfeldtKaski2024} gave an algebraic characterization of the existence of even directed cycles, a problem equivalent to Pfaffian recognition for bipartite graphs.

Considerable progress has also been made for some nonbipartite graph classes. Fischer and Little~\cite{FischerLittle2001} characterized Pfaffian near-bipartite graphs, which are nonbipartite matching covered graphs that become bipartite and remain matching covered after deleting two specific edges. Miranda and Lucchesi~\cite{MirandaLucchesi2010} later gave a polynomial time Pfaffian recognition algorithm for weakly near-bipartite graphs, a larger class that includes near-bipartite graphs. Carvalho, Lucchesi, and Murty~\cite{CarvalhoLucchesiMurty2012} extended Little's characterization to solid matching covered graphs, a class that includes bipartite matching covered graphs.

For general graphs, however, the problem has proved much more difficult. Norine and Thomas~\cite{NorineThomas2008} constructed an infinite family of minimal non-Pfaffian graphs with respect to matching minors, showing that the forbidden obstruction characterization for bipartite graphs does not extend directly to general graphs. Despite the progress on special graph classes, the existence of a deterministic polynomial time algorithm for recognizing general Pfaffian graphs remains open; see~\cite{Thomas2006,LucchesiMurty2024}.

A major insight into the Pfaffian recognition problem came from Vazirani and Yannakakis in 1989~\cite{VaziraniYannakakis1989}. Given a graph \(G\), they construct in polynomial time an orientation \(D^\star\) such that \(G\) admits a Pfaffian orientation if and only if \(D^\star\) is Pfaffian. Using this reduction, they also established that deciding whether a graph is Pfaffian belongs to \(\coNP\). Recall that a decision problem is in \(\coNP\) if every no instance has a polynomial-size certificate that can be verified in polynomial time. If \(G\) is not Pfaffian, then \(D^\star\) has two perfect matchings with different signs. Such a pair provides a polynomial-size certificate, since \(D^\star\) can be constructed in polynomial time, and we can verify that both are perfect matchings with opposite signs.

However, testing whether a given orientation is Pfaffian is not straightforward, since a graph may have exponentially many perfect matchings. In this paper, we give a randomized polynomial time algorithm for this problem using an algebraic approach, establishing the following result.

\begin{theorem}
\label{thm:main}
Pfaffian recognition is in \(\coRP\).
\end{theorem}

Recall that a decision problem belongs to \(\coRP\) if there is a randomized polynomial time algorithm that always accepts yes instances and rejects no instances with probability at least \(1/2\). For Pfaffian recognition, given a graph \(G\), we first use the reduction of Vazirani and Yannakakis to construct the orientation \(D^\star\). We then run our randomized test on \(D^\star\). If \(G\) is Pfaffian, then \(D^\star\) is Pfaffian and our test always accepts. If \(G\) is not Pfaffian, then \(D^\star\) is non-Pfaffian and our test rejects with probability at least \(1/2\). Repeating the test independently reduces the probability of error exponentially.

The proof of Theorem~\ref{thm:main} reduces testing whether a prescribed orientation is Pfaffian to polynomial identity testing over a finite field of characteristic two. Given a graph \(G\) with a perfect matching and an orientation \(D\), we fix a perfect matching and define a polynomial \(P_D(\mathbf{x})\) such that
\[
    D\text{ is Pfaffian} \quad\Longleftrightarrow\quad P_D(\mathbf{x})\equiv 0.
\]
The polynomial \(P_D\) records the perfect matchings whose sign is opposite to that of the fixed perfect matching. We express \(P_D\) in terms of determinants of principal submatrices of \(A_D\), allowing it to be evaluated in polynomial time. Applying the Schwartz--Zippel lemma gives a randomized test that always accepts Pfaffian orientations and rejects non-Pfaffian orientations with probability at least \(1/2\). We further show that each trial can be implemented using \(O(|V(G)|^3)\) field operations. We also show that a nonzero evaluation of \(P_D\) can be used to recover two perfect matchings of opposite signs in polynomial time.

A subgraph \(H\) of \(G\) is called \emph{conformal} if \(G-V(H)\) has a perfect matching. An even cycle is \emph{oddly oriented} if an odd number of its edges are directed along either direction of traversal, and \emph{evenly oriented} otherwise. If two perfect matchings have opposite signs, then their symmetric difference contains an evenly oriented conformal cycle. In fact, a classical characterization states that an orientation \(D\) is Pfaffian if and only if every conformal even cycle is oddly oriented~\cite{LovaszPlummer1986}.

We give a structural characterization of Pfaffian orientations in terms of strong connectivity. A directed graph is \emph{strongly connected} if every vertex can be reached from every other vertex by a directed path. For \(X\subseteq V(G)\), let \(\partial(X)\) denote the set of edges with exactly one endpoint in \(X\). Given an orientation \(D\) of \(G\), let \(D^X\) denote the orientation obtained by reversing every edge in \(\partial(X)\). Reversing the edges of a cut preserves Pfaffianity, that is, \(D\) is Pfaffian if and only if \(D^X\) is Pfaffian~\cite{CarvalhoLucchesiMurty2005}. For matching covered graphs, we establish the following characterization.

\begin{theorem}
\label{thm:structural}
Let \(G\) be a matching covered graph and let \(D\) be an orientation of \(G\). Then \(D\) is non-Pfaffian if and only if there exists \(X\subseteq V(G)\) such that \(D^X\) is strongly connected.
\end{theorem}

The proof of Theorem~\ref{thm:structural} proceeds as follows. If \(D\) is non-Pfaffian, then \(G\) has an evenly oriented conformal cycle in \(D\), which can be made directed by reversing the edges of a suitable cut. We then use the ear-decomposition theorem to inductively extend this directed cycle to a strongly connected orientation of the entire graph. At each step, the new ears can be made directed without changing the orientation of the previously constructed subgraph, thereby preserving strong connectivity. Conversely, if \(D^X\) is strongly connected, then every edge belongs to a directed cycle. This allows us to assign positive edge weights so that the total incoming and outgoing weights are equal at every vertex, making the weighted skew-adjacency matrix singular. If \(D\) were Pfaffian, then \(D^X\) would also be Pfaffian, since reversing the edges of a cut preserves Pfaffianity. Thus all perfect matchings would contribute with the same sign to the Pfaffian of the weighted skew-adjacency matrix. Since the graph has a perfect matching and all edge weights are positive, this Pfaffian would be nonzero. By Cayley's identity, the matrix would therefore be nonsingular, which gives a contradiction.

Interestingly, Klostermeyer~\cite{Klostermeyer1999} studied the problem of making a directed graph strongly connected by reversing the edges of a cut and showed that this problem is NP-complete for general digraphs. In contrast, by Theorem~\ref{thm:structural} and our randomized test for Pfaffian orientations, this problem admits a randomized polynomial time algorithm on matching covered graphs. For bipartite and planar matching covered graphs, it can be solved in deterministic polynomial time.

The rest of the paper is organized as follows. In Section~\ref{sec:wrong-sign}, we study the signs of perfect matchings and give an algebraic characterization of Pfaffian orientations. In Section~\ref{sec:randomized}, we use this characterization to obtain a randomized algorithm for Pfaffian recognition, give a faster implementation, and show how to find two perfect matchings of opposite signs. In Section~\ref{sec:structural}, we prove the characterization in terms of strong connectivity and discuss its algorithmic consequences.

\section{An algebraic characterization of Pfaffian orientations}
\label{sec:wrong-sign}

In this section, we construct a polynomial that is identically zero exactly when the prescribed orientation is Pfaffian.

\subsection{Signs of perfect matchings}

We use \(\mathbb F_2=\{0,1\}\) to record parity, with \(0\) representing even parity and \(1\) representing odd parity. For an edge \(e=ij\), with \(i<j\), say that \(e\) is \emph{backward} in
\(D\) if it is oriented from \(j\) to \(i\). For a perfect matching \(M\), let
\(\rev_D(M)\in\mathbb F_2\) be the parity of the number of backward edges in \(M\).

Two vertex disjoint edges \(\{i,j\}\) and \(\{k,\ell\}\), written with \(i<j\) and \(k<\ell\),
\emph{cross} if their endpoints alternate, that is, if \(i<k<j<\ell\) or \(k<i<\ell<j\). For a
perfect matching \(M\), let \(\crs(M)\in\mathbb F_2\) be the parity of the number of crossing pairs in \(M\). The following sign formula also appears in Norine~\cite{Norine2008}.

\begin{lemma}\label{lem:sign}
For every perfect matching \(M\),
\[
    \operatorname{sign}_D(M) = (-1)^{\crs(M)+\rev_D(M)}.
\]
\end{lemma}

\begin{proof}
For each edge of \(M\), write its endpoints as \((a_i,b_i)\) with \(a_i<b_i\), and order the edges so that \(a_1<a_2<\cdots<a_n\). Consider the permutation
\[
    \pi= \begin{pmatrix}
        1&2&\cdots&2n-1&2n\\
        a_1&b_1&\cdots&a_n&b_n
    \end{pmatrix}.
\]
Fix \(i<j\). Since \(a_i<a_j\), the only possible inversions involving the endpoints of the two edges \(\{a_i,b_i\}\) and \(\{a_j,b_j\}\) come from \(b_i\). If \(b_i<a_j\), there are no such inversions. If \(a_j<b_i<b_j\), there is exactly one inversion, while if \(b_j<b_i\), there are two. Thus the number of inversions contributed by this pair of edges is odd exactly when
\[
    a_i<a_j<b_i<b_j,
\]
which is exactly when the two edges cross. Therefore the parity of the total number of inversions of
\(\pi\) is \(\crs(M)\), and hence
\[
    \operatorname{sign}(\pi)=(-1)^{\crs(M)}.
\]
To obtain the oriented order defining \(\operatorname{sign}_D(M)\), we swap the two endpoints of each backward edge. Each such swap changes the sign by a factor of \(-1\). Therefore
\[
    \operatorname{sign}_D(M) = (-1)^{\crs(M)+\rev_D(M)}. \qedhere
\]
\end{proof}

\subsection{The wrong sign polynomial}

Assume that \(G\) has a perfect matching, and fix one such matching \(F\). Put
\(s=\crs(F)+\rev_D(F)\in\mathbb F_2\). By Lemma~\ref{lem:sign}, a perfect matching \(M\) has sign
opposite to \(F\) exactly when
\[
    s+\crs(M)+\rev_D(M)=1 \qquad\text{in }\mathbb F_2.
\]
Let $\x=(x_e)_{e\in E(G)}$, where the $x_e$ are variables over a finite field of characteristic two. Define
\[
    P_D(\x) = \sum_{\substack{M\in\cM(G)\\ \operatorname{sign}_D(M)\ne\operatorname{sign}_D(F)}}
        \prod_{e\in M}x_e^2.
\]
The polynomial depends on the choice of \(F\). Since we fix \(F\) throughout, we write \(P_D\) without mentioning \(F\) explicitly.

\begin{proposition}\label{prop:poly-characterization}
The orientation \(D\) is Pfaffian if and only if \(P_D(\x)\equiv0.\)
\end{proposition}

\begin{proof}
The monomials of \(P_D\) are in one to one correspondence with the perfect matchings whose sign differs from that of \(F\). Since \(F\) is itself a perfect matching, \(P_D\) is identically zero exactly when every perfect matching has the same sign as \(F\), which is exactly when \(D\) is Pfaffian.
\end{proof}

To obtain an efficient expression for \(P_D\), we use the weighted skew-adjacency matrix \(A_D\) over a field of characteristic two. For an edge \(e=ij\), let \(A_D-e\) denote the principal submatrix obtained by deleting rows and columns \(i,j\). For two vertex disjoint edges \(e=ij\) and \(f=k\ell\), let \(A_D-e-f\) denote the principal submatrix obtained by deleting rows and columns \(i,j,k,\ell\).

Let \(R_D\subseteq E(G)\) be the set of backward edges of \(D\), and let \(\mathcal C(G)\) be the set of unordered pairs \(\{e,f\}\) of vertex disjoint edges which cross in the fixed vertex ordering. The next identity expresses \(P_D\) in terms of principal minors of \(A_D\).

\begin{lemma}\label{lem:wrong-sign-poly}
Over every field of characteristic two,
\begin{equation}\label{eq:wrong-sign-poly}
    P_D(\x) = s\det(A_D) + \sum_{e\in R_D} x_e^2\det(A_D-e)
        + \sum_{\{e,f\}\in\mathcal C(G)} x_e^2x_f^2\det(A_D-e-f).
\end{equation}
\end{lemma}

\begin{proof}
In characteristic two, \(-1=1 \), so all terms in the Pfaffian expansion have positive signs. Thus
\[
    \Pf(A_D) = \sum_{M\in\mathcal M(G)} \prod_{e\in M}x_e.
\]
In characteristic two, squaring both sides gives
\[
    (\Pf(A_D))^2 = \sum_{M\in\mathcal M(G)} \prod_{e\in M}x_e^2.
\]
By Cayley's identity, \((\Pf A_D)^2=\det A_D\), and hence
\[
    \det(A_D) = \sum_{M\in\mathcal M(G)} \prod_{e\in M}x_e^2.
\]
The same argument applies to principal submatrices. Let \(e=ij\), and let \(G-\{i,j\}\) denote the graph obtained by deleting the vertices \(i\) and \(j\). Since \(A_D-e\) is the weighted skew-adjacency matrix of \(G-\{i,j\}\) with the induced orientation,
\[
    \det(A_D-e) = \sum_{M\in\mathcal M(G-\{i,j\})} \prod_{a\in M}x_a^2.
\]
Every perfect matching of \(G-\{i,j\}\) extends uniquely to a perfect matching of \(G\) containing
\(e\), so
\[
    x_e^2\det(A_D-e) = \sum_{\substack{M\in\mathcal M(G)\\ e\in M}} \prod_{a\in M}x_a^2.
\]
Likewise, for two vertex disjoint edges \(e=ij\) and \(f=k\ell\),
\[
    x_e^2x_f^2\det(A_D-e-f) = \sum_{\substack{M\in\mathcal M(G)\\ e,f\in M}} \prod_{a\in M}x_a^2.
\]
Fix a perfect matching \(M\) and consider the coefficient of \(\prod_{e\in M}x_e^2\) on the right hand side of \eqref{eq:wrong-sign-poly}. The term \(s\det A_D\) contributes \(s\). The second sum contributes once for each backward edge in \(M\), and the third once for each crossing pair in
\(M\). Its coefficient is therefore
\[
    s+\rev_D(M)+\crs(M).
\]
By Lemma~\ref{lem:sign}, this equals \(1\) exactly when
\(\operatorname{sign}_D(M)\ne\operatorname{sign}_D(F)\), and equals \(0\) otherwise. This is
precisely the coefficient of the same monomial in \(P_D(\x)\).
\end{proof}

\section{Randomized Pfaffian recognition}
\label{sec:randomized}

In this section, we use the polynomial \(P_D\) to obtain a randomized test for Pfaffian orientations and, using the reduction of Vazirani and Yannakakis, a randomized algorithm for Pfaffian graph recognition. We further show that each randomized trial can be implemented using \(O(|V(G)|^3)\) field operations. We also give an algorithm for finding a non-Pfaffian witness.

\subsection{Testing Pfaffian orientations}

We use the following standard form of the Schwartz-Zippel lemma~\cite{Schwartz1980,Zippel1979}.

\begin{lemma}\label{lem:schwartz-zippel}
Let \(P\in\mathbb F[x_1,\ldots,x_m]\) be a nonzero polynomial of total degree at most \(d\), and let
\(S\subseteq\mathbb F\) be finite. If \(\alpha_1,\ldots,\alpha_m\) are chosen independently and
uniformly from \(S\), then
\[
    \Pr[P(\alpha_1,\ldots,\alpha_m)=0] \le \frac{d}{|S|}.
\]
\end{lemma}

Applying this lemma to \(P_D\) gives the following randomized test.

\begin{theorem}\label{thm:testing}
There is a randomized polynomial time algorithm which, given a graph \(G\) and an orientation \(D\), always accepts if \(D\) is Pfaffian and rejects with probability at least \(1/2\) if \(D\) is non-Pfaffian.
\end{theorem}

\begin{proof}
If \(G\) has no vertices or no perfect matching, then \(D\) is Pfaffian, so we accept. Otherwise, find a perfect matching \(F\), write \(|V(G)|=2n\), and consider the polynomial \(P_D(\mathbf{x})\) defined above.

Choose a field \(\mathbb F_q\) of characteristic two, where \(q=2^t\) is the smallest power of \(2\) satisfying \(q\ge 4n\). Then \(4n\le q<8n\), and \(t=O(\log n)\), so field operations can be carried out in polynomial time. Such a field can be constructed as \(\mathbb F_2[z]/(h(z))\), where \(h\) is an irreducible polynomial of degree \(t\). Since there are only \(2^t=O(n)\) monic polynomials of degree \(t\), enumerating them and testing irreducibility takes deterministic polynomial time. The randomized algorithm proceeds as follows.

\begin{enumerate}
    \item For each edge \(e\in E(G)\), choose \(\alpha_e\in\mathbb F_q\) independently and uniformly at random.

    \item Evaluate \(P_D(\alpha)\) using Lemma~\ref{lem:wrong-sign-poly}.

    \item Reject if \(P_D(\alpha)\ne0\), and accept otherwise.
\end{enumerate}

If \(D\) is Pfaffian, then \(P_D(\x)\equiv0\) by Proposition~\ref{prop:poly-characterization}, so the algorithm always accepts. If \(D\) is non-Pfaffian, then \(P_D(\x)\) is a nonzero polynomial of degree \(2n\). By Lemma~\ref{lem:schwartz-zippel},
\[
    \Pr[P_D(\alpha)=0] \le \frac{2n}{q} \le \frac12.
\]
A non-Pfaffian orientation is therefore rejected with probability at least \(1/2\). By Lemma~\ref{lem:wrong-sign-poly}, evaluating \(P_D(\alpha)\) requires only polynomially many determinant computations on matrices of order at most \(2n\), so the algorithm runs in polynomial time.
\end{proof}

Repeating the test independently \(r\) times reduces the probability of incorrectly accepting a non-Pfaffian orientation to at most \(2^{-r}\), while Pfaffian orientations are always accepted.

\subsection{Pfaffian graph recognition}

We now prove Theorem~\ref{thm:main}. To obtain an algorithm for Pfaffian graph recognition from the test for a prescribed orientation, we use the following reduction of Vazirani and Yannakakis~\cite[Theorem~3.1]{VaziraniYannakakis1989}.

\begin{lemma}\label{lem:vy-reduction}
Let \(G\) be a graph with a perfect matching. There is a deterministic polynomial time algorithm that constructs an orientation \(D^\star\) of \(G\) such that \(G\) is Pfaffian if and only if \(D^\star\) is Pfaffian.
\end{lemma}

\begin{proof}[Proof of Theorem~\ref{thm:main}]
If \(G\) has no vertices or no perfect matching, then \(G\) is trivially Pfaffian, so we accept. Otherwise, construct \(D^\star\) using Lemma~\ref{lem:vy-reduction} and apply Theorem~\ref{thm:testing}. If \(G\) admits a Pfaffian orientation, then \(D^\star\) is Pfaffian, so the algorithm always accepts. If \(G\) does not admit a Pfaffian orientation, then \(D^\star\) is non-Pfaffian and is rejected with probability at least \(1/2\).
\end{proof}

\subsection{Efficient evaluation of the wrong sign polynomial}
\label{sec:evaluation}

The preceding algorithm evaluates \(P_D\) using polynomially many determinant computations. We now show that, when \(A_D(\alpha)\) is invertible, \(P_D(\alpha)\) can be evaluated using \(O(|V(G)|^3)\) field operations. The main idea is to compute the required principal minors from a single matrix inverse.

We use Jacobi's complementary minor identity \cite{HornJohnson2013}. For an invertible matrix \(A\) indexed by \(V\) and a subset \(X\subseteq V\),
\[
    \det(A[V\setminus X]) = \det(A)\det(A^{-1}[X]),
\]
where \(A[X]\) denotes the principal submatrix of \(A\) with rows and columns indexed by \(X\).

In our application, the principal submatrices are obtained by deleting two or four rows and columns. Jacobi's identity therefore expresses their determinants in terms of principal submatrices of the inverse of order at most four.

Let \(\alpha=(\alpha_e)_{e\in E(G)}\) be a point over a field of characteristic two, and write \(A_\alpha=A_D(\alpha)\). Set \(\alpha_{ij}=0\) whenever \(ij\notin E(G)\). Suppose that \(A_\alpha\) is invertible, and write \(B=A_\alpha^{-1}=(b_{ij})_{i,j\in V(G)}\).

We first show that both \(A_\alpha\) and \(B\) are alternating. Over a field of characteristic two, an \emph{alternating matrix} is a symmetric matrix with zero diagonal. For any alternating matrix \(A=(a_{ij})_{n\times n}\) and any vector \(v\in\mathbb F^n\), we have
\[
    v^TAv = \sum_{i,j=1}^n a_{ij}v_iv_j
    = \sum_{i=1}^n a_{ii}v_i^2
    + \sum_{1\le i<j\le n}(a_{ij}+a_{ji})v_iv_j=0,
\]
since \(a_{ii}=0\) and \(a_{ij}=-a_{ji}\). By definition, \(A_\alpha\) is alternating. In characteristic two, symmetric and skew-symmetric matrices are the same, so \(A_\alpha^T=A_\alpha\). Hence
\[
    (A_\alpha^{-1})^T=(A_\alpha^T)^{-1}=A_\alpha^{-1}.
\]
Thus \(A_\alpha^{-1}\) is also symmetric. To see that \(A_\alpha^{-1}\) has zero diagonal, let \(e_i\) be the vector of length \(|V(G)|\) with \(1\) in position \(i\) and \(0\) elsewhere. Then
\[
    b_{ii}=e_i^T A_\alpha^{-1}e_i
    =(A_\alpha^{-1}e_i)^T A_\alpha(A_\alpha^{-1}e_i)=0,
\]
where the last equality follows since \(A_\alpha\) is alternating. Hence \(A_\alpha^{-1}\) has zero diagonal and is therefore alternating. For an edge \(e=ij\), Jacobi's identity gives
\begin{equation}\label{eq:single-minor}
    \begin{aligned}
        \det(A_\alpha-e)
        &=\det(A_\alpha)\det\bigl(A_\alpha^{-1}[\{i,j\}]\bigr)\\
        &=\det(A_\alpha)
        \det\left(
            \begin{smallmatrix}
                0 & b_{ij}\\
                b_{ij} & 0
            \end{smallmatrix}
        \right)
        =\det(A_\alpha)b_{ij}^2.
    \end{aligned}
\end{equation}
Similarly, for two vertex disjoint edges \(e=ij\) and \(f=k\ell\), Jacobi's identity gives
\begin{equation}\label{eq:double-minor}
    \begin{aligned}
        \det(A_\alpha-e-f)
        &=\det(A_\alpha)
        \det\left(
            \begin{smallmatrix}
                0 & b_{ij} & b_{ik} & b_{i\ell}\\
                b_{ij} & 0 & b_{jk} & b_{j\ell}\\
                b_{ik} & b_{jk} & 0 & b_{k\ell}\\
                b_{i\ell} & b_{j\ell} & b_{k\ell} & 0
            \end{smallmatrix}
        \right)\\
        &=\det(A_\alpha)
        \bigl(b_{ij}b_{k\ell}+b_{ik}b_{j\ell}
        +b_{i\ell}b_{jk}\bigr)^2.
    \end{aligned}
\end{equation}
The last equality follows since the determinant of an alternating matrix of even order is the square of its Pfaffian. Define
\begin{equation}\label{eq:Q}
    \begin{aligned}
        Q_D(\alpha)={}&s+\sum_{ij\in R_D}\alpha_{ij}b_{ij} +\sum_{\{ij,k\ell\}\in\mathcal C(G)}
        \alpha_{ij}\alpha_{k\ell}
        \bigl(b_{ij}b_{k\ell}+b_{ik}b_{j\ell}
        +b_{i\ell}b_{jk}\bigr).
    \end{aligned}
\end{equation}

\begin{lemma}\label{lem:one-inverse}
If \(A_\alpha\) is invertible, then \(P_D(\alpha)=\det(A_\alpha)Q_D(\alpha)^2\).
\end{lemma}

\begin{proof}
Substituting \eqref{eq:single-minor} and \eqref{eq:double-minor} into \eqref{eq:wrong-sign-poly}, we obtain
\[
    \begin{aligned}
        P_D(\alpha)={}&s\det(A_\alpha)
        +\sum_{ij\in R_D}\alpha_{ij}^2\det(A_\alpha)b_{ij}^2
        +\sum_{\{ij,k\ell\}\in\mathcal C(G)}
        \alpha_{ij}^2\alpha_{k\ell}^2\det(A_\alpha)
        \bigl(b_{ij}b_{k\ell}+b_{ik}b_{j\ell}
        +b_{i\ell}b_{jk}\bigr)^2.
    \end{aligned}
\]
Taking \(\det(A_\alpha)\) common gives
\[
    \begin{aligned}
        P_D(\alpha)=\det(A_\alpha)\Biggl(
        &s+\sum_{ij\in R_D}(\alpha_{ij}b_{ij})^2
        +\sum_{\{ij,k\ell\}\in\mathcal C(G)}
        \bigl(\alpha_{ij}\alpha_{k\ell}
        (b_{ij}b_{k\ell}+b_{ik}b_{j\ell}
        +b_{i\ell}b_{jk})\bigr)^2
        \Biggr).
    \end{aligned}
\]
Since squaring distributes over sums in characteristic two and \(s^2=s\), the expression in parentheses is exactly \(Q_D(\alpha)^2\) by \eqref{eq:Q}. Therefore, \(P_D(\alpha)=\det(A_\alpha)Q_D(\alpha)^2\).
\end{proof}

\begin{lemma}\label{lem:cubic-evaluation}
If \(A_\alpha\) is invertible, then \(P_D(\alpha)\) can be evaluated using \(O(|V(G)|^3)\) field operations.
\end{lemma}

\begin{proof}
The determinant and inverse of \(A_\alpha\) can be computed by Gaussian elimination using \(O(|V(G)|^3)\) field operations. By Lemma~\ref{lem:one-inverse}, it is enough to show that \(Q_D(\alpha)\) can be computed within the same bound. The constant term and the sum over \(R_D\) take only \(O(|V(G)|^2)\) operations. The only nontrivial part is the sum over crossing pairs,
\[
    S=\sum_{\{ij,k\ell\}\in\mathcal C(G)}
    \alpha_{ij}\alpha_{k\ell}
    \bigl(b_{ij}b_{k\ell}+b_{ik}b_{j\ell}
    +b_{i\ell}b_{jk}\bigr).
\]
Every pair of crossing edges can be written uniquely as \(ij,k\ell\), where \(i<k<j<\ell\). A direct computation of \(S\) would therefore use four nested loops, one for each of \(i,k,j,\ell\), and would take \(O(|V(G)|^4)\) operations. We show how to avoid the loop over \(i\). Using the order \(i<k<j<\ell\), we can rewrite \(S\) as
\[
    S=\sum_j\sum_{k<j}\sum_{\ell>j}\sum_{i<k}
    \alpha_{ij}\alpha_{k\ell}
    \bigl(b_{ij}b_{k\ell}+b_{ik}b_{j\ell}
    +b_{i\ell}b_{jk}\bigr).
\]
For fixed \(j,k,\ell\), the factors which do not depend on \(i\) can be taken outside the innermost sum. Hence
\[
    \begin{aligned}
        S=\sum_j\sum_{k<j}\sum_{\ell>j}
        \alpha_{k\ell}\Bigg(
        &b_{k\ell}\sum_{i<k}\alpha_{ij}b_{ij}
+b_{j\ell}\sum_{i<k}\alpha_{ij}b_{ik}
        +b_{jk}\sum_{i<k}\alpha_{ij}b_{i\ell}
        \Bigg).
    \end{aligned}
\]
For fixed \(j\), the inner sums can be updated efficiently as \(k\) increases. Set \(S=0\). For each \(j\), perform the following steps.

\begin{enumerate}
    \item Set \(h_v\leftarrow 0\) for every vertex \(v\).

    \item Process \(k<j\) in increasing order. For each such \(k\):
    \begin{enumerate}
        \item For each \(\ell>j\), set
        \[
            S\leftarrow S+
            \alpha_{k\ell}
            \bigl(b_{k\ell}h_j+b_{j\ell}h_k+b_{jk}h_\ell\bigr).
        \]

        \item After all \(\ell>j\) have been processed, set
        \(
            h_v\leftarrow h_v+\alpha_{kj}b_{kv}
        \)
        for every vertex \(v\).
    \end{enumerate}
\end{enumerate}

We now show that the above procedure computes exactly \(S\). Before processing a given \(k\), the value \(h_v\) is
\[
    h_v=\sum_{i<k}\alpha_{ij}b_{iv}.
\]
Thus \(h_j\), \(h_k\), and \(h_\ell\) are exactly the three inner sums appearing in the expression for \(S\). Therefore the update of \(S\) adds the contribution of all choices \(i<k\) for the fixed \(j,k,\ell\). After processing \(k\), the update \(h_v\leftarrow h_v+\alpha_{kj}b_{kv}\) adds the new term corresponding to \(i=k\). Hence the same identity for \(h_v\) holds for the next value of \(k\). So, as \(j\), \(k\), and \(\ell\) run through all possible choices, we add exactly all terms with \(i<k<j<\ell\). These are exactly the terms in the sum defining \(S\). Hence the procedure computes \(S\).

We now count the number of field operations. For each fixed pair \(j,k\), the loop over \(\ell>j\) takes \(O(|V(G)|)\) field operations. Updating \(h_v\) for all vertices \(v\) also takes \(O(|V(G)|)\) field operations. Since there are \(O(|V(G)|^2)\) choices of \(j,k\), the total number of field operations is
\[
    O(|V(G)|^2)\cdot O(|V(G)|)
    =O(|V(G)|^3).
\]
Thus \(S\), and hence \(Q_D(\alpha)\), can be computed using \(O(|V(G)|^3)\) field operations. By Lemma~\ref{lem:one-inverse}, the same bound holds for evaluating \(P_D(\alpha)\).
\end{proof}

The preceding lemma applies only when \(A_\alpha\) is invertible. In the randomized test, we accept whenever \(A_\alpha\) is singular. Since \(G\) has a perfect matching, \(\det(A_D)\) is a nonzero polynomial, so the probability of obtaining a singular matrix can be bounded using the Schwartz--Zippel lemma. This gives the following result.

\begin{corollary}\label{cor:cubic-testing}
Let \(G\) be a graph with a perfect matching, and let \(D\) be an orientation of \(G\). There is a randomized test for whether \(D\) is Pfaffian using \(O(|V(G)|^3)\) field operations per trial. The test always accepts if \(D\) is Pfaffian and rejects with probability at least \(1/2\) otherwise.
\end{corollary}

\begin{proof}
If \(G\) has no vertices or no perfect matching, then \(D\) is trivially Pfaffian, so we accept. Otherwise, let \(|V(G)|=2n\), and let \(q=2^t\) be the smallest power of two satisfying \(q\ge 8n\). Choose \(\alpha\in\mathbb F_q^{E(G)}\) uniformly at random. We use the following test.
\begin{enumerate}
    \item If \(A_\alpha\) is singular, accept.

    \item Otherwise, compute \(Q_D(\alpha)\). Reject if \(Q_D(\alpha)\ne0\), and accept otherwise.
\end{enumerate}

If \(D\) is Pfaffian, then \(P_D(\alpha)=0\) for every \(\alpha\). If \(\det(A_\alpha)=0\), the test accepts. Otherwise, \(A_\alpha\) is invertible, and by Lemma~\ref{lem:one-inverse},
\[
    P_D(\alpha)=\det(A_\alpha)Q_D(\alpha)^2=0.
\]
Since \(\det(A_\alpha)\ne0\), we have \(Q_D(\alpha)=0\), so the test also accepts in this case. Hence the test always accepts when \(D\) is Pfaffian.

Now suppose that \(D\) is non-Pfaffian. Then \(P_D\) is a nonzero polynomial of degree \(2n\). Since \(G\) has a perfect matching, \(\det(A_D)\) is also a nonzero polynomial of degree \(2n\). The test can accept only if \(P_D(\alpha)=0\) or \(\det(A_\alpha)=0\). By the Schwartz--Zippel lemma and the union bound,
\[
    \begin{aligned}
        \Pr\bigl[P_D(\alpha)=0\text{ or }\det(A_\alpha)=0\bigr]
        &\le \Pr[P_D(\alpha)=0]+\Pr[\det(A_\alpha)=0]\\
        &\le \frac{2n}{q}+\frac{2n}{q}
        =\frac{4n}{q}\le\frac12.
    \end{aligned}
\]
Thus, with probability at least \(1/2\), \(A_\alpha\) is invertible and \(P_D(\alpha)\ne0\). By Lemma~\ref{lem:one-inverse}, this implies \(Q_D(\alpha)\ne0\), and the test detects that \(D\) is non-Pfaffian.

Testing whether \(A_\alpha\) is singular takes \(O(|V(G)|^3)\) field operations. If it is invertible, its inverse and \(Q_D(\alpha)\) can also be computed within the same bound by Lemma~\ref{lem:cubic-evaluation}. Hence each trial uses \(O(|V(G)|^3)\) field operations.
\end{proof}

\subsection{Finding a certificate of non-Pfaffianity}

The randomized test detects non-Pfaffian orientations but does not directly produce a certificate. We now show how a nonzero evaluation of \(P_D\) can be used to find two perfect matchings of opposite signs by self-reduction.

\begin{theorem}\label{thm:witness-recovery}
There is a randomized polynomial time algorithm that, given a non-Pfaffian orientation \(D\) of a graph \(G\), finds two perfect matchings of opposite signs with probability at least \(1/2\).
\end{theorem}

\begin{proof}
We use the randomized test of Theorem~\ref{thm:testing} to find a perfect matching \(F\) and an assignment \(\alpha\) such that \(P_D(\alpha)\ne0\). If the test accepts, return no certificate. Otherwise, we recover a perfect matching of sign opposite to \(F\) as follows.

\begin{enumerate}
    \item Set \(\beta=\alpha\) and \(M=\varnothing\).

    \item Choose a vertex \(v\) not covered by \(M\). For each edge \(e\) incident with \(v\), let \(\beta^{(e)}\) agree with \(\beta\), except that the weights of all other edges incident with \(v\) are set to zero.

    \item Find an edge \(e=uv\) such that \(P_D(\beta^{(e)})\ne0\). Replace \(\beta\) by \(\beta^{(e)}\), set to zero the weights of all other edges incident with \(u\), and add \(e\) to \(M\).

    \item Repeat Steps 2 and 3 until every vertex is covered by \(M\). Return \(F\) and \(M\).
\end{enumerate}

We show that the algorithm maintains \(P_D(\beta)\ne0\) throughout the construction. Since every perfect matching contains exactly one edge incident with \(v\), we have
\[
P_D(\beta)=\sum_{e\ni v}P_D(\beta^{(e)}).
\]
As \(P_D(\beta)\ne0\), there exists an edge \(e=uv\) such that \(P_D(\beta^{(e)})\ne0\). Since all edges incident with previously covered vertices, except the selected edges, have weight zero, the vertex \(u\) must also be uncovered. Moreover, every perfect matching contributing a nonzero term to \(P_D(\beta^{(e)})\) contains \(e\). Thus setting the weights of all other edges incident with \(u\) to zero does not change the evaluation. Hence the selected edges remain disjoint and \(P_D(\beta)\ne0\) is preserved.

Since \(P_D(\beta)\ne0\), there exists a perfect matching containing all selected edges. When every vertex has been covered, \(M\) is the only perfect matching that can contribute a nonzero term to \(P_D(\beta)\). Therefore, \(M\) has sign opposite to \(F\), as required.

The algorithm requires at most \(O(|E(G)|)\) evaluations of \(P_D\), each of which can be performed in polynomial time by Lemma~\ref{lem:wrong-sign-poly}. By Theorem~\ref{thm:testing}, a nonzero initial evaluation is obtained with probability at least \(1/2\), completing the proof.
\end{proof}

\section{A structural characterization of Pfaffian orientations}
\label{sec:structural}

In this section, we prove Theorem~\ref{thm:structural}, which characterizes non-Pfaffian orientations of matching covered graphs in terms of strong connectivity, and discuss its algorithmic consequences. We first recall the classical characterization of Pfaffian orientations in terms of conformal even cycles~\cite[Theorem~8.3.2]{LovaszPlummer1986}.

\begin{lemma}\label{lem:conformal-cycle}
Let \(G\) be a graph with a perfect matching and let \(D\) be an orientation of \(G\). Then \(D\) is Pfaffian if and only if every conformal even cycle of \(G\) is oddly oriented in \(D\).
\end{lemma}

We also use the ear-decomposition theorem for matching covered graphs. An \emph{odd ear} of a subgraph \(H\) is a path of odd length whose endpoints belong to \(V(H)\), whose internal vertices lie outside \(V(H)\), and whose edges do not belong to \(H\). An \emph{ear decomposition} of a matching covered graph \(G\) is a sequence of matching covered subgraphs starting with \(K_2\) and ending with \(G\), where each successive graph is obtained from the previous one by adding either one odd ear or two vertex disjoint odd ears. Lovász and Plummer showed that an ear decomposition can be extended from any conformal matching covered subgraph~\cite{LovaszPlummer1986}. We use the following form of their theorem~\cite[Theorem~11.11]{LucchesiMurty2024}.

\begin{lemma}\label{lem:ear-decomposition}
Let \(G\) be a matching covered graph and let \(H\) be a conformal matching covered subgraph of \(G\). Then there exists a sequence
\[
H=H_0\subseteq H_1\subseteq\cdots\subseteq H_t=G
\]
of matching covered subgraphs such that each \(H_{i+1}\) is obtained from \(H_i\) by adding either one odd ear or two vertex disjoint odd ears.
\end{lemma}

We finally recall that reversing the edges of a cut preserves Pfaffianity~\cite{CarvalhoLucchesiMurty2005}.

\begin{lemma}\label{lem:cut-reversal}
Let \(G\) be a graph with a perfect matching, \(D\) an orientation of \(G\), and \(X\subseteq V(G)\). Then \(D\) is Pfaffian if and only if \(D^X\) is Pfaffian.
\end{lemma}

To see this, every perfect matching contains a number of edges of \(\partial(X)\) having the same parity as \(|X|\). Thus, reversing these edges changes the signs of all perfect matchings by the same factor \((-1)^{|X|}\). We can now prove Theorem~\ref{thm:structural}.

\begin{proof}[Proof of Theorem~\ref{thm:structural}]
We first prove the forward implication. Suppose that \(D\) is non-Pfaffian. By Lemma~\ref{lem:conformal-cycle}, there exists an evenly oriented conformal cycle \(C\).

We first show that \(C\) can be made directed by reversing a suitable cut. Write \(C=v_1v_2\cdots v_{2k}v_1\) and fix this direction of traversal. Starting with \(X_0=\varnothing\), we consider the vertices \(v_2,v_3,\ldots,v_{2k}\) in order. At each vertex \(v_i\), if the edge \(v_{i-1}v_i\) points against the traversal, we switch \(v_i\) and add it to \(X_0\). Here switching a vertex means reversing all edges incident with it. In this way, every edge except possibly \(v_{2k}v_1\) points along the traversal. Each switch reverses exactly two edges of \(C\), so it preserves the parity of the number of edges directed against the traversal. Since \(C\) was evenly oriented, this number remains even. Thus the last edge must also point along the traversal, and \(C\) becomes a directed cycle. Since the vertices in \(X_0\) have each been switched exactly once, the resulting orientation is \(D^{X_0}\), as an edge is reversed precisely when exactly one of its endpoints belongs to \(X_0\).

Since \(C\) is a conformal matching covered subgraph of \(G\), by Lemma~\ref{lem:ear-decomposition}, there exists an ear decomposition starting with \(C\):
\[
C=H_0\subseteq H_1\subseteq\cdots\subseteq H_t=G,
\]
where each step adds either one odd ear or two vertex disjoint odd ears. We show by induction that, for each \(i\), there exists \(X_i\subseteq V(H_i)\) such that the restriction of \(D^{X_i}\) to \(H_i\) is strongly connected. This holds for \(i=0\), since \(C\) is directed in \(D^{X_0}\).

Suppose that the restriction of \(D^{X_i}\) to \(H_i\) is strongly connected. Let \(P\) be one of the one or two odd ears added in passing from \(H_i\) to \(H_{i+1}\). We show that each such ear can be made directed by switching only its internal vertices, without changing the orientation of \(H_i\).

Write \(P=v_0v_1\cdots v_{2k+1}\) and choose a direction of traversal such that an even number of its edges point against the traversal in \(D^{X_i}\). Such a direction always exists because \(P\) has odd length, so reversing the direction of traversal changes the parity of the number of edges directed against the traversal. Starting with \(Y_i=\varnothing\), we consider the internal vertices \(v_1,v_2,\ldots,v_{2k}\) in order. At each vertex \(v_j\), if the edge \(v_{j-1}v_j\) points against the traversal in the current orientation, we switch \(v_j\) and add it to \(Y_i\). In this way, every edge except possibly \(v_{2k}v_{2k+1}\) points along the traversal. Each switch reverses exactly two edges of \(P\), so the parity of the number of edges directed against the traversal remains even. Thus the last edge must also point along the traversal, and \(P\) becomes a directed path in \(D^{X_i\cup Y_i}\). Since only internal vertices were switched, neither endpoint was switched. This also covers an ear of length one, for which no switches are needed.

If two odd ears are added, we apply the same construction independently to both, taking \(Y_i\) to be the union of the two sets of switched vertices. Since all these vertices lie outside \(V(H_i)\), switching them does not change the orientation of any edge of \(H_i\).

Set \(X_{i+1}=X_i\cup Y_i\). We show that the restriction of \(D^{X_{i+1}}\) to \(H_{i+1}\) is strongly connected. The orientation of \(H_i\) remains unchanged, and every new vertex lies on a directed path whose endpoints belong to \(H_i\). From any vertex of \(H_i\), we can reach a new vertex by first reaching the initial endpoint of its ear and then following the directed path. Similarly, from any new vertex, we can reach every vertex of \(H_i\) by following the rest of its ear and then a directed path within \(H_i\). Since \(H_i\) is strongly connected, every pair of vertices of \(H_{i+1}\) is connected by directed paths in both directions. This completes the induction. Finally, \(H_t=G\), so \(D^{X_t}\) is strongly connected. This proves the forward implication.

For the converse, suppose that \(D^X\) is strongly connected for some \(X\subseteq V(G)\). We show that \(D\) is non-Pfaffian.

Since \(D^X\) is strongly connected, every directed edge belongs to a directed cycle. For each edge \(f\in E(G)\), choose a directed cycle \(C_f\) containing \(f\), and assign to each edge \(e\) the weight
\[
w_e=\bigl|\{f\in E(G):e\in E(C_f)\}\bigr|.
\]
Since \(e\in E(C_e)\), we have \(w_e\geq 1\) for every edge \(e\). Now fix a vertex \(v\). Each chosen directed cycle either avoids \(v\) or contains exactly one edge entering \(v\) and one edge leaving \(v\). Since \(w_e\) counts how many chosen cycles contain \(e\), summing over all these cycles gives the same total weight to the incoming and outgoing edges at \(v\). Let \(A=A_{D^X}(w)\) be the weighted skew-adjacency matrix. For every vertex \(v\in V(G)\), we have
\[
(A\mathbf{1})_v
=\sum_{\substack{e\text{ outgoing}\\\text{from }v}}w_e
-\sum_{\substack{e\text{ incoming}\\\text{to }v}}w_e
=0.
\]
Thus \(A\mathbf{1}=0\), so \(A\) is singular and \(\det A=0\).

Suppose that \(D\) is Pfaffian. By Lemma~\ref{lem:cut-reversal}, \(D^X\) is also Pfaffian. Therefore,
\[
\operatorname{Pf}(A)
=\pm\sum_{M\in\mathcal M(G)}
\prod_{e\in M}w_e.
\]
Since \(G\) has a perfect matching and all weights are positive, the sum on the right is strictly positive, so \(\operatorname{Pf}(A)\ne0\). By Cayley's identity,
\[
\det A=\operatorname{Pf}(A)^2>0,
\]
contradicting \(\det A=0\). Therefore \(D\) is non-Pfaffian.
\end{proof}

This gives the following algorithmic consequence.

\begin{corollary}\label{cor:strong-pushing}
Given an orientation \(D\) of a matching covered graph \(G\), there is a randomized polynomial time algorithm to decide whether some \(X\subseteq V(G)\) makes \(D^X\) strongly connected. If \(G\) is bipartite or planar, the problem can be solved in deterministic polynomial time.
\end{corollary}

\begin{proof}
By Theorem~\ref{thm:structural}, such a set \(X\) exists if and only if \(D\) is non-Pfaffian. Thus Theorem~\ref{thm:testing} gives a randomized polynomial time algorithm.

For bipartite graphs, deterministic polynomial time algorithms for testing Pfaffian orientations are known~\cite{RST1999,McCuaig2004}. For planar graphs, the FKT algorithm computes the number of perfect matchings in polynomial time~\cite{Kasteleyn1963}. Comparing this number with \(\lvert\operatorname{Pf}(A_D(\mathbf{1}))\rvert\) determines whether \(D\) is Pfaffian. Hence the problem is solvable in deterministic polynomial time in both cases.
\end{proof}

\section{Open questions}

The main open question is whether Pfaffian graph recognition admits a deterministic polynomial time algorithm. One possible approach is to determine whether the determinant representation of the wrong sign polynomial in \eqref{eq:wrong-sign-poly} admits deterministic polynomial identity testing. Another approach is to use our structural characterization in Theorem~\ref{thm:structural} to decide deterministically whether an orientation of a matching covered graph can be made strongly connected by reversing a cut. It would also be interesting to determine whether the algebraic approach of Section~\ref{sec:wrong-sign} can be applied directly to graph recognition, without first constructing the orientation of Vazirani and Yannakakis.

\section*{AI Disclosure}
The mathematical results were obtained independently by the authors. ChatGPT (OpenAI) was used only for grammar and language corrections in the manuscript. The authors take full responsibility for the correctness of the results and references.

\bibliographystyle{alphaurl}
\bibliography{references}

\end{document}